\documentclass[a4paper]{article}
\usepackage{graphicx}
\usepackage{graphics}
\usepackage{amssymb}
\usepackage{amsthm}
\usepackage{geometry}
\usepackage{txfonts}
\usepackage{mathdots}
\usepackage{natbib}
\usepackage{mathrsfs}

\newtheorem{theorem}{Theorem}[section]
\newtheorem{remark}[theorem]{Remark}

\begin{document}
\begin{center}
{\LARGE \textbf{Some notes on biasedness and unbiasedness of
two-sample Kolmogorov-Smirnov test}}\\
P. Bubeliny\\
\textsc{e-mail: bubeliny@karlin.mff.cuni.cz}\\
Charles University, Faculty  of  Mathematics  and
Physics, KPMS,  Sokolovska 83, Prague, Czech Republic, 18675.\\
\end{center}

\textbf{Abstract:}\textit{ This paper deals with two-sample
Kolmogorov-Smirnov test and its biasedness. This test is not
unbiased in general in case of different sample sizes. We found out
most biased distribution for some values of significance level
$\alpha$. Moreover we discovered that there exists number of
observation and significance level $\alpha$ such that this test is
unbiased at level
$\alpha$.}\\

\section{Introduction}
\indent In the world of statistic, there exists an enormous number
of tests and new ones are going to be derived. For most of these
tests we know, that they are consistent, we know their asymptotic
behavior and a lot of another properties. But there is one thing
which is often omitted. This thing is unbiasedness.\\
\indent Somebody can think, that all of the tests, which are used,
are unbiased or are biased against very special alternative which
can not occur in practical applications. Somebody can look at
unbiasedness as at very poor power of tests against some
alternatives and somebody can just thing that unbiasedness is
unimportant. But they are all wrong. We often check some assumptions
of test by other tests. But what if the checking test is biased and
therefore it leads to the bad decision? Then the main test should
not be used and it can lead to
wrong decision. Therefore, unbiasedness should not be underestimate.\\
\indent There are a lot of tests which are really unbiased. But
there are plenty of tests that are used daily and they are biased.
One of such tests is well known two-sample Kolmogorov-Smirnov test.
In what follows, we look at biasedness and unbiasedness of this test
in some cases in detail.

\section{Biasedness and unbiasedness of Kolmogorov-Smirnov test}

\indent Firstly, we should recall, what unbiasedness is. A test is
said to be unbiased at level $\alpha$ if
\begin{enumerate}
\item it has significance level $\alpha$
\item for all distributions from alternative the power of this test
is greater or equal to $\alpha$.
\end{enumerate}
The test is said to be unbiased if it is unbiased at all level
$\alpha\in(0,1)$. Finally, the test is said to be biased if it is
not unbiased. Specially, the test is biased at level $\alpha$
against alternative $G$ if it is an level $\alpha$ test and $P(\textrm{reject} H|G)<\alpha$.\\
\indent Consider, that $x_1,\dots,x_n$ and $y_1,\dots,y_m$ are two
independent samples having distributions with continuous
distribution functions $F$ and $G$, respectively. We would like to
test the hypothesis $H:$ $F=G$ against the alternative $A:$ $F\neq
G$. Then two-sample Kolmogorov-Smirnov test is based on statistic
\[
D_{n,m}=\sup_x|\hat{F}_n(x)-\hat{G}_m(x)|,
\]
where $\hat{F}_n(x)$ and $\hat{G}_m(x)$ are empirical distribution
functions of $F$ and $G$. The hypothesis $H$ is rejected for large
value of $D_{n,m}$. The exact formula for computing
$p$-values can be found in \cite{csen}.\\
\indent Firstly, we should realize that statistic $D_{n,m}$ of
two-sample Kolmogorov-Smirnov test has discrete distribution.
Therefore $p$-values for this test are discrete as well. For example
consider that $n=m=50$. Then the test statistic $D_{n,m}$ can take
just 50 different values $1/n,2/n,\dots,1$. For statistic
$D_{n,m}=0.26$ the $p$-value is equal to $0.0678$ and for the next
value $D_{n,m}=0.28$ the $p$-value is equal to $0.0392$. Testing at
level $\alpha=0.05$ could be little bit confusing because the power
of this test is equal for each value $\alpha\in [0.0392,0.0678)$.
There exists distribution $G$ such that power of Kolmogorov-Smirnov
test at level $\alpha=0.05$ is equal to $0.045$. Such distribution
does not meet requirements of definition of unbiasedness for
$\alpha=0.05$ though the power of this test is higher than exact
level of this test equal to $0.0392$. To hold the idea of
unbiasedness for tests with discrete test statistic we should
consider just discrete values of significance level $\alpha$ or use
randomized versions of these tests.\\
\indent It should be kept in mind that Kolmogorov-Smirnov test does
not depend on monotonic transformation of samples. If we transform
both samples (by the same monotonic transformation) to samples with
distribution functions $F'$ and $G'$, respectively then
$\sup_x|\hat{F}_n(x)-\hat{G}_m(x)|=\sup_x|\hat{F}_n'(x)-\hat{G}_m'(x)|$.
Therefore without loss of generality, we assume that $F$ is
distribution function of uniform distribution given by
\begin{equation}\label{unif}
F(x) =  \left\{
\begin{array}{ll}
0 & \textrm{if $x<0$}\\
x & \textrm{if $0\leq x \leq 1$}\\
1 & \textrm{if $x > 1$}\\
\end{array}. \right.
\end{equation}

\indent In \cite{cGK}, they proved that for $n=m$ there exist
$\alpha\in(0,1)$ such that two-sample Kolmogorov-Smirnov test is
unbiased at level $\alpha$ against two-sided alternative $F\neq G$.
If we consider just one-sided alternatives $A_1: F\leq G$ or $A_2:
F\geq G$ we can extend this founding to $n\neq m$.
\begin{theorem}
Let $x_1,\dots,x_n$ and $y_1,\dots,y_m$ be independent samples from
distribution $F$ and $G$. Then for arbitrary $n,m\in N$, there
exists $\alpha\in(0,1)$ such that two-sample Kolmogorov-Smirnov test
of hypothesis $H:\,F=G$ against one-sided alternative $A_1: F\leq G$
or $A_2: F\geq G$ is unbiased at level $\alpha$.
\end{theorem}
\begin{proof}
Without loss of generality, we consider that the first sample
$x_1,\dots,x_n$ is from uniform distribution.\\
Firstly, we consider only the alternative $A_1: F\leq G$. For this
alternative, the Kolmogorov-Smirnov statistic is given by
$D^*_{n,m}=\sup_{x \in (0,1)}\big(\hat{F}_n(x)-\hat{G}_m(x)\big)$,
where $\hat{F}_n$ and $\hat{G}_m$ are empirical distribution
functions of $F$ and $G$. The hypothesis $H$ is rejected for small
values of $D^*_{n,m}$. Consider $\alpha$ such small, that we reject
hypotheses $H$ for $D_{n,m}$ equals to minus one. It occurs if and
only if the samples $x_1,\dots,x_n$ and $y_1,\dots,y_m$ satisfy
\begin{equation}\label{maxmin}
\max(y_1,\dots,y_m)<\min(x_1,\dots,x_n).
\end{equation}
The probability of this event is given by
\begin{equation}\label{int1}
n\int_0^1(1-x)^{n-1} G^m(x) dx.
\end{equation}
Moreover, $G(x)$ is monotone and $G(x)\geq x$ because we consider
alternative $A_1: F\leq G$. Therefore the function $(1-x)^{n-1}
G^m(x)$ of integral (\ref{int1}) attains its minimum for $G(x)=x$.
This integral represents probability of rejection of hypothesis at
level $\alpha$ if alternative $G$ is true and it is minimized for
$F=x=G(x)$. Hence, Kolmogorov-Smirnov test is unbiased at level
$\alpha$.\\
\indent The proof for alternative $A_2: F\geq G$ is similar. We take
$\alpha$ such small, that we reject hypothesis if and only if
$D_{n,m}=1$. The inequality (\ref{maxmin}) change to
\[
\max(x_1,\dots,x_n)<\min(y_1,\dots,y_m)
\]
and probability of this event is then given by
\begin{equation}\label{int2}
n\int_0^1x^{n-1} (1-G(x))^m dx
\end{equation}
For alternative $A_2$, we have $G(x)\leq x$ and hence integral
(\ref{int2}) is minimized for $G(x)=x$. It proves the theorem.
\end{proof}
\indent The result of this theorem does not mean that two-sample
Kolmogorov-Smirnov test is unbiased against one-sided alternative.
It only says that there exist small level $\alpha$ for which this
test is unbiased. In the following theorem we show that for $n\neq
m$ two-sided Kolmogorov-Smirnov test is not unbiased against
two-sided alternative.
\begin{theorem}
Let $x_1,\dots,x_n$ be i.i.d from uniform distribution with
distribution function $F$ and $y_1,\dots,y_m$ be i.i.d. from
distribution having distribution function $G$. If $n\neq m$ then
there exists $\alpha\in(0,1)$ such that two-sample
Kolmogorov-Smirnov test of hypothesis $H:F=G$ is biased against
alternative with the distribution function
\begin{equation}\label{dfa2}
G(x)=\frac{(\frac{x}{1-x})^{\frac{n-1}{m-1}}}{1+(\frac{x}{1-x})^{\frac{n-1}{m-1}}}.
\end{equation}
\end{theorem}
\begin{proof}
Consider $\alpha$ such small, that we reject hypotheses if and only
if $D_{n,m}=\sup_x|\hat{F}_n(x)-\hat{G}_m(x)|$ is equal to one. That
is, the samples $x_1,\dots,x_n$ and $y_1,\dots,y_m$ have to satisfy
\begin{equation}\label{emm}
\max(y_1,\dots,y_m)<\min(x_1,\dots,x_n)\,\, \textrm{or}\,\,
\max(x_1,\dots,x_n)<\min(y_1,\dots,y_m).
\end{equation}
The probability of this event is given by
\[
n\int_0^1\big((1-x)^{n-1} G^m(x) +x^{n-1} (1-G(x))^m\big)\, dx.
\]
Substitute $G(x)$ by $y$ and let the derivative of function
$(1-x)^{n-1} y^m +x^{n-1} (1-y)^m$ according to $y$ equal to zero.
It leads to the equation
\[
\big(\frac{y}{1-y}\big)^{m-1}=\big(\frac{x}{1-y}\big)^{n-1}.
\]
Therefore the probability of event (\ref{emm}) is not minimized for
$F(x)=G(x)=x$ but for
\[
G(x)=\frac{(\frac{x}{1-x})^{\frac{n-1}{m-1}}}{1+(\frac{x}{1-x})^{\frac{n-1}{m-1}}}.
\]
\end{proof}
\begin{figure}[h]
  \begin{center}
    \includegraphics[width=14cm]{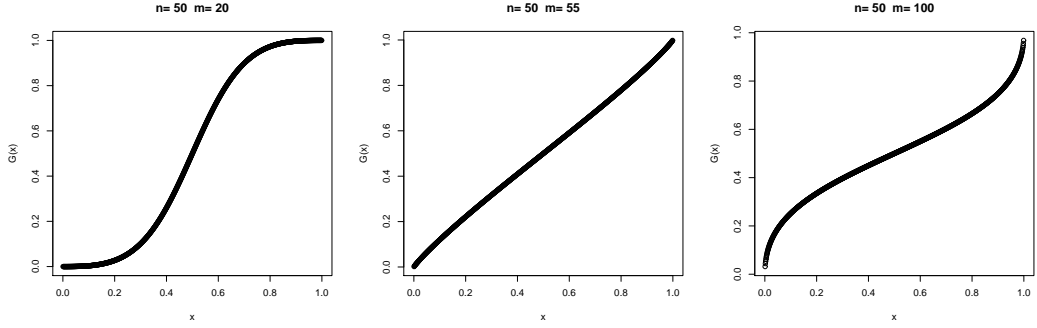}
    \caption{Plot of distribution function $G$ given by (\ref{dfa2}) for $n=50$ and $m=20,55,100$}
     \label{fdfa2}
  \end{center}
\end{figure}
Some examples of distribution function given by (\ref{dfa2}) are in
figure \ref{fdfa2}. Although we found out that two-sample
Kolmogorov-Smirnov test is biased against alternative (\ref{dfa2})
it is really true for very small $\alpha$. Let denote this smallest
level $\alpha$ by $\alpha_1$. Then $\alpha_1$ can be directly
computed by
\begin{equation}\label{ealpha1}
\alpha_1=n\int_0^1(1-x)^{n-1} x^m +x^{n-1} (1-x)^m
\,dx=2nm\frac{\Gamma(n)\Gamma(m)}{\Gamma(n+m+1)}.
\end{equation}
For example if $n=10$ and $m=11$ then $\alpha_1$ is equal to $5.67$x$10^{-6}$.\\
\indent All previous result are considered for Kolmogorov-Smirnov
statistic $D_{n,m}=1$. Let consider second highest value of this
statistic. For $n>m$ it is equal to $1-1/n$ and for $n<m$ it is
equal to $1-1/m$, respectively. We denote by $\alpha_2$ the
significance level $\alpha$ such that we reject two-sample
Kolmogorov-Smirnov test if and only if
$D_{n,m}\geq \max(1-1/n,1-1/m)$.\\

\indent Firstly, assume that $n>m\geq2$ and consider that
$D_{n,m}=1-1/n$. It can occur if and only if these samples are such
that $x_{(1)}<\dots<x_{(n-1)}<y_{(1)}<x_{(n)}$ or
$x_{(1)}<y_{(m)}<x_{(2)},\dots<x_{(n)}$. Together with the case
$D_{n,m}=1$ ($x_{(n)}<y_{(1)}$ or $y_{(m)}<x_{(1)}$) we have that
$D_{n,m}$ is greater or equal to $1-1/n$ if and only if
$x_{(n-1)}<y_{(1)}$ or $y_{(m)}<x_{(2)}$. It leads as to the
probability of rejecting the hypotheses at level $\alpha_2$
\begin{eqnarray}\label{ednmb}
P(D_{n,m}\geq 1-1/n)&=&P(\forall_j\, y_j>x_{(n-1)})+P(\forall_j\,
y_j<x_{(2)})\nonumber\\
&=& n(n-1)\int_0^1 \big(x^{n-2}(1-x)(1-G(x))^m+x(1-x)^{n-2}
G^m(x)\big)\,dx.
\end{eqnarray}
As in proof of  previous theorem let $G(x)=y$ and let the derivative
of interior function of integral (\ref{ednmb}) according to $y$
equal to zero. It leads us to solve the equation
\[
\big(\frac{y}{1-y}\big)^{m-1}=\big(\frac{x}{1-x}\big)^{n-3}.
\]
The solution $y$ as a function of $x$ is given by
\begin{equation}\label{esoldnmb}
y=G(x)=\frac{(\frac{x}{1-x})^{\frac{n-3}{m-1}}}{1+(\frac{x}{1-x})^{\frac{n-3}{m-1}}}.
\end{equation}
\newline
\indent Now assume that $2\leq n<m$ and consider $D_{n,m}=1-1/m$. It
can be true if and only if $y_{(1)}<\dots<y_{(m-1)}<x_{(1)}<y_{(m)}$
or $y_{(1)}<x_{(n)}<y_{(2)},\dots<y_{(m)}$. Therefore the
probability of event $D_{n,m}\geq 1-1/m$ is equal to
\begin{eqnarray}\label{ednma}
P(D_{n,m}\geq 1-1/m)&=&P(D_{n,m}=1-1/m)+P(D_{n,m}=1)\nonumber\\
&=&nm\int_0^1\big((1-x)^{n-1}G^{m-1}(x)(1-G(x))+x^{n-1}(1-G(x))^{m-1}G(x)\big)\,
dx\nonumber\\
&&+n\int_0^1\big((1-x)^{n-1} G^m(x) +x^{n-1} (1-G(x))^m\big)\, dx.
\end{eqnarray}
As before let $G(x)=y$ and let the derivative of interior function
of integral (\ref{ednma}) according to $y$ equal to zero. It leads
us to the equation
\[
\big(\frac{y}{1-y}\big)^{m-3}=\big(\frac{x}{1-x}\big)^{n-1}.
\]
Therefore the distribution function of most biased distribution in
this case is given by
\begin{equation}\label{esoldnma}
y=G(x)=\frac{(\frac{x}{1-x})^{\frac{n-1}{m-3}}}{1+(\frac{x}{1-x})^{\frac{n-1}{m-3}}}.
\end{equation}
\begin{remark}
If $n=3$ and $m=2$ or $n=2$ and $m=3$ then the most biased
distribution is discrete distribution given by probabilities
$P(y=0)=P(y=1)=\frac{1}{2}$ or $P(y=\frac{1}{2})=1$, respectively.
\end{remark}
Consider $G(x)=x$ then level $\alpha_2$ is given (according to
(\ref{ednmb}) and (\ref{ednma})) by
\begin{equation}\label{ealpha2}
\alpha_2=2nmk\frac{\Gamma(n)\Gamma(m)}{\Gamma(n+m+1)}=k\alpha_1,
\end{equation}
where $k=\min(n+1,m+1)$. Distribution functions (\ref{esoldnmb}) and
(\ref{esoldnma}) are similar to $S$-curves on figure \ref{fdfa2}.
Although these distribution functions are not equal to themselves
and to (\ref{dfa2}) as well, some interesting results can be found.
If $|n-m|=2$ then (\ref{esoldnmb}) and (\ref{esoldnma}) change to
$G(x)=x$. It means that the distribution which minimize
(\ref{ednmb}) and (\ref{ednma}) is uniform distribution. It leads us
to the following theorem.
\begin{theorem}\label{thnm2}
Let $\alpha_{n,m}$ be given by ($\ref{ealpha2}$). If $n=m+2$ or
$n=m-2$ then two-sample Kolmogorov-Smirnov test is unbiased at level
$\alpha_{n,m}$. Moreover, if $n\neq m$ and $|n-m|\neq 2$ then
Kolmogorov-Smirnov test is biased at level $\alpha_{n,m}$.
\end{theorem}
\begin{proof}
Because of $\alpha_{n,m}$=$\alpha_2$, the most biased distribution
functions are given by (\ref{esoldnmb}) and (\ref{esoldnma}). For
$|n-m|=2$ they change to $G(x)=x=F(x)$. It means that the uniform
distribution minimize the probability of rejection hypotheses $F=G$
against alternative $F\neq G$ at level $\alpha_2$ if and only if
$|n-m|=2$.
\end{proof}
\begin{remark}\label{rknm1}
If $|n-m|=1$ then Kolmogorov-Smirnov test is not biased against the
distribution functions (\ref{esoldnmb}) and (\ref{esoldnma}) at
level $\alpha_1$.
\end{remark}
Let denote by $\mathscr{A}_\alpha$ the set of distributions for
which Kolmogorov-Smirnov test is biased at level $\alpha$, it is
\[
\mathscr{A_\alpha}=\{G: P(\textrm{reject $H$ at level }\alpha
|\textrm{alternative }G\textrm{ is true})<\alpha\}.
\]
For different levels $0<\alpha<\alpha^*$, one would expect that
there is some subset relation between $\mathscr{A_\alpha}$ and
$\mathscr{A_{\alpha^*}}$. But it is not generally true. According to
the theorem \ref{thnm2} there exist $G_\alpha$ such that
$G_\alpha\in\mathscr{A_\alpha}$ and
$G_\alpha\notin\mathscr{A_{\alpha^*}}$. On the other hand, from
remark \ref{rknm1} we have that there exists $G_\alpha^*$ such that
$G_\alpha^*\notin\mathscr{A_\alpha}$ and
$G_\alpha^*\in\mathscr{A_{\alpha^*}}$. Therefore, in general
$\mathscr{A_\alpha}$ is not subset of $\mathscr{A_{\alpha^*}}$ and
vice versa.\\

\indent The previous result can be quite simply generalized to
$\alpha_3$ (the third smallest $\alpha$) in case of $n>2m$ or
$2n<m$. Adding the probability of the even $D_{n,m}=1-2/m$ or
$D_{n,m}=1-2/n$ to the (\ref{ednmb}) or (\ref{ednma}) leads us to
the most biased distributions at level $\alpha_3$ given by
\begin{equation}\label{ednm3}
G_3(x) =
\frac{(\frac{x}{1-x})^{\frac{n-5}{m-1}}}{1+(\frac{x}{1-x})^{\frac{n-5}{m-1}}} \hspace{1cm} \textrm{if $n>2m$}\\
\end{equation}
or
\begin{equation}\label{ednm3}
G_3(x) = \frac{(\frac{x}{1-x})^{\frac{n-1}{m-5}}}{1+(\frac{x}{1-x})^{\frac{n-1}{m-5}}} \hspace{1cm} \textrm{if $m>2n$}.\\
\end{equation}
In this case, $\alpha_3$ is given by
\[
\alpha_3=2k_2nm\frac{\Gamma(n)\Gamma(m)}{\Gamma(n+m+1)}=k_2\alpha_1,
\]
where $k_2=\frac{\min((m+2)(m+1),(n+2)(n+1))}{2}$. If $n=m+4$ or
$m=n+4$ then $G_3(x)=x$. Together with condition $n>2m$ or $m>2n$ we
have that for $n=6,m=2$ or $n=2,m=6$ the two-sample
Kolmogorov-Smirnov test is unbiased at level $\alpha_3=3/7$ and for
$n=7,m=3$ or $n=3,m=7$ the two-sample Kolmogorov-Smirnov test is
unbiased at
level $\alpha_3=1/6$.\\

\indent Sofar considered $\alpha$'s are too small in case we have
some tens of observation in each sample.  Therefore we perform the
following simulation to look if two-sample Kolmogorov-Smirnov test
is biased against the distribution ($\ref{dfa2}$) at level
$\alpha\approx0.05$. We set the number of observation $n$ for the
first sample be $n=10,20,50,100$ and the number of observation $m$
for the second sample be $m=11,15,21,51,101$. As a distribution of
the first sample we consider uniform distribution and for second
sample we consider two distributions. The first one is the uniform
distribution and the second one is distribution having distribution
function $G$ given by (\ref{dfa2}). We perform 10000 repetitions and
compute the difference between the estimate of power if second
sample is from alternative distribution and the estimated level
$\alpha$ if the second sample is from uniform distribution. The
results of this simulation are in table \ref{tabPOW}. We can see
that for all considered $n$ and $m$ the estimate of difference is
greater than $0$. It means that two-sample Kolmogorov-Smirnov test
is not biased against
alternative (\ref{dfa2}) at level $\alpha=0.05$.\\
\begin{table}[ht]
\begin{center}
\caption{Difference between estimate of power for alternative $G$
given by (\ref{dfa2}) and estimate of level $\alpha$ of two-sample
Kolmogorov-Smirnov test.}\label{tabPOW}
\begin{tabular}{|c|rrrrr|}
  \hline
$\alpha=5\%$&m=11&m=15&m=21&m=51&m=101\\
  \hline
$n=10$& 0.0034 & 0.0144 & 0.0320 & 0.4153 & 0.7290 \\
$n=20$& 0.0291 & 0.0087 & 0.0016 & 0.2784 & 0.9170 \\
$n=50$& 0.4071 & 0.3403 & 0.2715 & 0.0001 & 0.5291 \\
$n=100$& 0.9070 & 0.9189 & 0.9190 & 0.4557 & 0.0001 \\
 \hline
\end{tabular}
\end{center}
\end{table}

\section{Conclusion}

\indent In this paper we looked at biasedness and unbiasedness of
two-sample Kolmogorov-Smirnov test. In case of different sample
sizes this test is not unbiased. However we found out that it is not
true for all $\alpha \in (0,1)$. There exists some special
combination of number of observations in each sample and
significance level $\alpha$ at which this test is unbiased (see
 e.g theorem {\ref{thnm2}). Moreover, we discovered the most biased
distribution for some values of $\alpha$. Although we consider just
small values of $\alpha$, for small sample sizes or for data such as
gene expressions these levels of $\alpha$ are appropriate. We did
not consider all levels of $\alpha$. However we point out that this
test can be unbiased for large samples and $\alpha$ around 0.05.
However more research is needed to find out the exact relation
between number of observations and level $\alpha$ at which this test
is unbiased.

\section*{Acknowledgments}
\indent The author thanks Prof. Lev Klebanov, DrSc. for valuable
comments, remarks and overall help. The work was supported by the
grant SVV 261315/2011.

\end{document}